\documentclass[a4,11pt]{article}
\usepackage{multirow}

\newtheorem{lemma}{Lemma}
\newtheorem{theorem}{Theorem}%[section]
\newenvironment{proof}%
{\begin{trivlist}\item[\hspace*{\labelsep}{\it Proof.\/}]}%
{\hfill$\Box$\end{trivlist}}

\newcommand{\head}[1]
 {\markright{\hbox to 0pt{\vtop to 0pt{\hbox{}\vskip 3mm \hrule
 width  \textwidth \vss} \hss}{\sc #1}}}
\usepackage{latexsym,graphicx,amsfonts,amsmath}
\begin{document}

\title{A note on scheduling with low rank processing times}
\author{Lin Chen \and Deshi Ye
\and Guochuan Zhang}
\date{College of Computer Science, Zhejiang
University, Hangzhou, 310027, China}
\maketitle
\baselineskip 15pt

\begin{abstract}
We consider the classical minimum makespan scheduling problem, where
the processing time of job $j$ on machine $i$ is $p_{ij}$, and the
matrix $P=(p_{ij})_{m\times n}$ is of a low rank. It is proved in
~\cite{low} that rank 7 scheduling is NP-hard to approximate to a
factor of $3/2-\epsilon$, and rank 4 scheduling is APX-hard (NP-hard
to approximate within a factor of $1.03-\epsilon$). We improve this
result by showing that rank 4 scheduling is already NP-hard to
approximate within a factor of $3/2-\epsilon$, and meanwhile rank 3 scheduling
is APX-hard.

\vskip 2mm\noindent{\bf Keywords.}{Complexity, APX-hardness, Scheduling}
\end{abstract}
\section{Introduction}
Recently Bhaskara et al.~\cite{low} study the minimum makespan
scheduling problem in which the processing time of job $j$ on
machine $i$ is $p_{ij}$, and the matrix formed by the processing
times $P=(p_{ij})_{m\times n}$ is of a low rank. Formally speaking,
in this problem the matrix of processing times could be expressed as
$P=MJ$, where $M$ is an $m\times D$ matrix in which the row vector
$u_i$ represents the D-dimensional {\em speed vector} of machine
$i$, and $J$ is a $D \times n$ matrix in which the column vector
$v_j^T$ represents the D-dimensional {\em size vector} of job $j$.
The processing time of job $j$ on machine $i$ is defined by
$u_i\cdot v_j^T$. We adopt the notations of~\cite{low} by denoting
the above problem as $LRS(D)$. It is easy to see that in $LRS(D)$,
the rank of the matrix $P$ is at most $D$.

It is a new way of studying the traditional scheduling problem. From
this point of view the unrelated machine scheduling problem is a
scheduling problem where the matrix of job processing times could be
of arbitrary rank, while the related machine scheduling problem is a
scheduling problem where the matrix is of rank one.

In 1988, Hochbaum and Shmoys~\cite{HocShm88} gave a PTAS (polynomial
time approximation scheme) for the related machine scheduling
problem, i.e., $LRS(1)$. Later, Lenstra et al.~\cite{lenstra1990unr}
provided a 2-approximation algorithm for the unrelated machine
scheduling problem, i.e., $LRS(D)$ for arbitrary $D$. Such a result
was improved to a $(2 - 1/m)$-approximation algorithm by Shchepin
and Vakhania~\cite{shchepin2005optimal}. It remains open whether
there exists a polynomial time algorithm with approximation ratio
strictly less than $2$ for the unrelated machine scheduling problem.

Bhaskara et al.~\cite{low} prove that $LRS(D)$ is APX-hard (NP-hard
to approximate within a factor of $1.03-\epsilon$) when $D=4$, and
NP-hard to approximate within a factor of $3/2-\epsilon$ when $D=7$.

In this paper we improve the results in~\cite{low} by showing that
$LRS(4)$ is already NP-hard to approximate within a factor of
$3/2-\epsilon$, and $LRS(3)$ is APX-hard.

Roughly speaking, the overall structure of our reduction for $LRS(4)$ is similar
to that of~\cite{low}. The key ingredient to our stronger result is
that we construct the reduction from a variation of the
3-dimensional matching problem (instead of the standard
3-dimensional matching problem, as is used in~\cite{low}), and
design the job processing times in a more delicate way. For the APX-hardness of
$LRS(3)$, we make use of the idea from
~\cite{ETH} to design the reduction from the one-in-three 3SAT problem.

\section{Inapproximability of Rank 4 scheduling}
In this section, we prove that $LRS(4)$ is already NP-hard to approximate within a factor
of $3/2-\epsilon$ for any $\epsilon>0$ via a reduction from a variation of the 3 dimensional matching
problem, as is shown in the following.

\subsection{A variation of the 3-Dimensional Matching problem}\label{sec:3dmp}
The standard 3DM problem contains three disjoint element sets $W\cup
X\cup Y$ where $|W|=|X|=|Y|$, and a set of triples
$T\subseteq\{(w_i,x_j,y_k)|w_i\in W, x_j\in X, y_k\in Y\}$ where
every triple of $T$ is called as a {\em match}. A {\em perfect
matching} for 3DM is a subset $T'\subseteq T$ in which every element
of $W\cup X\cup Y$ appears exactly once. Deciding whether there
exists a perfect matching for 3DM is NP-complete~\cite{GarJoh79}.

In the standard 3DM, the subscripts of elements in a match could be
arbitrary. In this paper, however, we focus on the following
restricted form of 3DM.

\begin{itemize}
\item Elements: there are three disjoint sets of elements $W=\{w_i,\bar{w}_i|i=1,\cdots,3n\}$,
$X=\{s_i,a_i|i=1,\cdots,3n\}$ and $Y=\{s_i',b_i|i=1,\cdots,3n\}$
where $|W|=|X|=|Y|=6n$
\item Matches: there are two sets of matches
$T_1\subseteq\{(w_i,s_j,s_j'),(\bar{w}_i,s_j,s_j')|w_i\in W,s_j\in
X,s_j'\in Y\}$,
$T_2=\{(w_i,a_i,b_i),(\bar{w}_i,a_i,b_{\zeta{(i)}})|i=1,\cdots,3n\}$
where $\zeta$ is defined as $\zeta(3k+1)=3k+2$, $\zeta(3k+2)=3k+3$
and $\zeta(3k+3)=3k+1$ for $k=0,\cdots,n-1$
\end{itemize}

We remark that in the special form of 3DM above, $T_2$ is already
fixed. Similarly, a subset of $T_1\cup T_2$ is called a perfect
matching if among its matches every element of $W\cup X\cup Y$
appears exactly once. For simplicity, we denote the problem of
determining whether the set of matches in the above form admits a
perfect matching as 3DM$^\prime$. We prove that the 3DM$^\prime$
problem is NP-complete in the following theorem. It turns out that
the idea of the proof is similar to that of~\cite{ETH}.

\begin{theorem}
3DM$^\prime$ is NP-complete.
\end{theorem}
\begin{proof}
%\subsection{3DM' is NP-hard}
We reduce 3SAT to 3DM$^\prime$. Given an instance of 3SAT, say,
$I_{sat}$, we first apply Tovey's method~\cite{tovey1984simplified}
to alter it into $I_{sat}'$ so that every variable appears exactly
three times. It is simple, if a variable, say, $z$, only appear
once, then we add a dummy clause $(z\vee\neg z)$. Otherwise it
appears $d\ge 2$ times, and we replace all its occurrences by new
variables $z_1$, $z_2$, $\cdots$, $z_d$, one for each, and meanwhile
add clauses $(z_1\vee\neg z_2)$, $(z_2\vee\neg z_3)$, $\cdots$,
$(z_{d}\vee\neg z_1)$ to enforce that $z_1$ to $z_d$ are taking the
same truth value. Let $I_{sat}'$ be the new SAT instance, then we
have
\begin{itemize}
\item Every variable appears exactly three times in $I_{sat}'$.
\item $I_{sat}'$
is satisfiable if and only if $I_{sat}$ is satisfiable.
\end{itemize}

We apply Tovey's method for a second time to transform $I_{sat}'$ to
$I_{sat}''$. Since every variable, say, $z_i$, appears three times
in $I_{sat}'$, we replace its three occurrences with
$\hat{z}_{3i-2}$, $\hat{z}_{3i-1}$ and $\hat{z}_{3i}$, and meanwhile
add $(\hat{z}_{3i-2}\vee\neg \hat{z}_{3i-1})$,
$(\hat{z}_{3i-1}\vee\neg \hat{z}_{3i})$, $(\hat{z}_{3i}\vee\neg
\hat{z}_{3i-2})$. It is not difficult to verify that $I_{sat}''$
satisfies the following conditions:

\begin{itemize}
\item Clauses of $I_{sat}''$ could be divided into $C_1$ and $C_2$
such that
\begin{itemize}
\item either $z_i$ or $\neg z_i$ appears in $C_1$, and it appears
once
\item all the clauses of $C_2$ could be listed as
$(z_{3i-2}\vee\neg z_{3i-1})$, $(z_{3i-1}\vee\neg z_{3i})$,
$(z_{3i}\vee\neg z_{3i-2})$ for $i=1,\cdots,n$.
\end{itemize}
\item $I_{sat}''$
is satisfiable if and only if $I_{sat}'$ is satisfiable.
\end{itemize}

We construct an instance $I_{3dm}$ (of 3DM$^\prime$) such that it
admits a perfect matching if and only if $I_{sat}''$ is satisfiable,
and thus if and only if $I_{sat}$ is satisfiable.

\noindent\textbf{Real elements:} we construct $w_i\in W$ for every
positive literal $z_i$, and $\bar{w}_i\in W$ for negative literal
$\neg z_i$. We construct $s_j\in X$ and $s_j'\in Y$ for every clause
$\beta_j\in C_1$.

\noindent\textbf{Dummy elements:} we construct $a_i\in X,b_i\in Y$
for $1\le i\le 3n$, and $u_j\in X$, $u_j'\in Y$ for $1\le j\le
3n-|C_1|$ (here $|C_1|$ is the number of clauses in $C_1$). Thus in
all, $|W|=|X|=|Y|=6n$.

\noindent\textbf{Real matches:} if the positive literal $z_i$ is in
clause $\beta_j\in C_1$, we construct $(w_i,s_j,s_j')$. Else if $\neg
z_i$ is in $\beta_j$, we construct $(\bar{w}_i,s_j,s_j')$.

\noindent\textbf{Dummy matches:} we construct $(w_i,a_i,b_i)$ and
$(\bar{w}_i,a_i,b_{\zeta{(i)}})$ for $i=1,\cdots,3n$. We also
construct $(w_i,u_j,u_j')$ and $(\bar{w}_i,u_j,u_j')$ for every
$1\le i\le 3n$ and $1\le j\le 3n-|C_1|$.

It is easy to see that the above instance is an instance of
3DM$^\prime$ where $W=\{w_i,\bar{w}_i|1\le i\le 3n\}$,
$X=\{a_i,s_j,u_k|1\le i\le 3n, 1\le j\le |C_1|, 1\le k\le
3n-|C_1|\}$, $Y=\{b_i,s_j',u_k'|1\le i\le 3n, 1\le j\le |C_1|, 1\le
k\le 3n-|C_1|\}$, $T_2$ is a subset of dummy matches (i.e.,
$(w_i,a_i,b_i)$ and $(\bar{w}_i,a_i,b_{\zeta{(i)}})$), and $T_1$ is
the set of remaining dummy matches (i.e., $(w_i,u_j,u_j')$ and
$(\bar{w}_i,u_j,u_j')$) together with all the real matches.

\noindent\textbf{Completeness.} Suppose $I_{sat}''$ is satisfiable,
we prove that $I_{3dm}$ admits a perfect matching by selecting them
out from $T_1\cup T_2$. Since clauses $(z_{3i+1}\vee\neg z_{3i+2})$,
$(z_{3i+2}\vee\neg z_{3i+3})$ and $(z_{3i+3}\vee\neg z_{3i})$ are
all satisfied for $0\le i\le n-1$, the variables $z_{3i+1}$,
$z_{3i+2}$, $z_{3i+3}$ should be all true or all false. If they are
all false, we take out $(w_{3i+1},a_{3i+1},b_{3i+1})$,
$(w_{3i+2},a_{3i+2},b_{3i+2})$ and $(w_{3i+3},a_{3i+3},b_{3i+3})$
from $T_2$. Otherwise we take out
$(\bar{w}_{3i+1},a_{3i+1},b_{3i+2})$,
$(\bar{w}_{3i+2},a_{3i+2},b_{3i+3})$,
$(\bar{w}_{3i+3},a_{3i+3},b_{3i+1})$ from $T_2$ instead. Since each
clause $\beta_j\in C_1$ is satisfied, it is satisfied by at least one
literal in it. Suppose it is satisfied by the literal $z_i$ (or
$\neg z_i$), then the variable $z_i$ is true (or false), and we
select $(w_i,s_j,s_j')$ (or $(\bar{w}_i,s_j,s_j')$) from $T_1$. Now
consider all the matches we select out so far. We have selected
$3n+|C_1|$ matches, and among them every $s_j$, $s_j'$, $a_i$, $b_i$
appear once, and every $w_i$ (or $\bar{w}_i$) appears at most once.
Thus, there are $3n-|C_1|$ elements of $W$ that do not appear in
these matches, and we select $3n-|C_1|$ dummy matches from $T_1$ so
that every $u_j$, $u_j'$ and $w_i$, $\bar{w}_i$ appear once.

\noindent\textbf{Soundness.} Suppose there exists a perfect matching
of $I_{3dm}$, say, $T'\subseteq T_1\cup T_2$, we prove that
$I_{sat}''$ is satisfiable.

Consider elements of $X$ and $Y$. For each $0\le i\le n-1$, to
ensure that $a_{3i+1}$, $b_{3i+1}$, $a_{3i+2}$, $b_{3i+2}$ and
$a_{3i+3}$, $b_{3i+3}$ appear once respectively, in the perfect
matching $T'$ we have to choose either
$(\bar{w}_{3i+1},a_{3i+1},b_{3i+2})$,
$(\bar{w}_{3i+2},a_{3i+2},b_{3i+3})$,
$(\bar{w}_{3i+3},a_{3i+3},b_{3i+1})$, or choose
$(w_{3i+1},a_{3i+1},b_{3i+1})$, $(w_{3i+2},a_{3i+2},b_{3i+2})$,
$(w_{3i+3},a_{3i+3},b_{3i+3})$.

If $(\bar{w}_{3i+1},a_{3i+1},b_{3i+2})$,
$(\bar{w}_{3i+2},a_{3i+2},b_{3i+3})$,
$(\bar{w}_{3i+3},a_{3i+3},b_{3i+1})$ are in $T'$, we let variables
$z_{3i+1}$, $z_{3i+2}$ and $z_{3i+3}$ be true. Otherwise we let
variables $z_{3i+1}$, $z_{3i+2}$ and $z_{3i+3}$ be false. It can be
easily verified that every clause of $C_2$ is satisfied.

We consider $\beta_j\in C_1$. Notice that $s_j\in X$ appears once in
$T'$. Suppose the match containing $s_j$ is $(w_i,s_j,s_j')$ for
some $i$, then it follows that the positive literal $z_i\in \beta_j$.
The fact that $w_i$ can only appear once in $T'$ implies that
$(w_i,a_i,b_i)$ is not in $T'$, and thus the variable $z_i$ is true
and $\beta_j$ is satisfied. Otherwise the matching containing $s_j$ is
$(\bar{w}_i,s_j,s_j')$ for some $i$, and similar arguments show that
the negative literal $\neg z_i\in \beta_j$ and variable $z_i$ is false,
again $\beta_j$ is satisfied.
\end{proof}

\subsection{Factor 3/2 hardness for rank 4 scheduling}\label{sec:r4}
We prove the following theorem.
\begin{theorem}
For every $\epsilon>0$, there is no $(3/2-\epsilon)$-approximation
algorithm for $LRS(4)$, assuming $P\neq NP$.
\end{theorem}
\begin{proof}
Given an instance $I_{3dm}$ of 3DM$^\prime$, we construct an
instance $I_{sch}$ of $LRS(4)$ such that if $I_{3dm}$ admits a
perfect matching, then there exists a feasible schedule of $I_{sch}$
with makespan $2+O(\epsilon)$ (where $\epsilon<1/6$ is an arbitrary
positive number), otherwise there is no feasible schedule of
makespan less than $3$.

We construct $I_{sch}$ that consists of the following parts.
\begin{itemize}
\item Machines: there are $|T|=|T_1\cup T_2|$ machines, one for every match.
\item Real jobs: there is one job for each element of $X\cup Y$.
\item Dummy jobs: if $w_i$ (or $\bar{w}_i$)
appears $d(w_i)$ (or $d(\bar{w}_i)$) times in all the matches, then
there are $d(w_i)-1$ (or $d(\bar{w}_i)-1$) jobs for $w_i$ (or
$\bar{w}_i$).
\end{itemize}

Let $N=O(n/\epsilon^2)$. We aim to design the speeds of machines and
the size (workload) of jobs such that if an element job is put on a
match machine whose corresponding match contains this element, then
its processing time is $1+O(\epsilon)$, otherwise the processing
time is at least $1/(2\epsilon)>3$.

See Table~\ref{table:machine} as the speed vectors of machines, and
Table~\ref{table:job} as the size vector of jobs. Here for
simplicity we use a match to denote its corresponding machine, and
an element to denote its corresponding job. Recall that the
processing time of a job on a machine is defined to be the inner
product of their corresponding vectors.
\begin{table}[!hbp]
\begin{center}
\caption{Speed Vectors of Machines}
\begin{tabular}{|c|c|}
\hline Machines & Speeds\\
\hline $(w_i,s_j,s_j')$ & $(N^i,N^{-i},N^{j+N},N^{-j-N})$\\
\hline $(\bar{w}_i,s_j,s_j')$ & $(N^{-i},N^{i},N^{j+N},N^{-j-N})$\\
\hline $(w_{3i},a_{3i},b_{3i})$ & $(N^{3i},N^{-3i},N^{-3i},N^{-3i-1})$ \\
\hline $(w_{3i+1},a_{3i+1},b_{3i+1})$ & $(N^{3i+1},N^{-3i-1},N^{-3i-1},2N^{-3i})$ \\
\hline $(w_{3i+2},a_{3i+2},b_{3i+2})$ & $(N^{3i+2},N^{-3i-2},N^{-3i-2},2/\epsilon N^{-3i-1})$ \\
\hline $(\bar{w}_{3i},a_{3i},b_{3i+1})$ & $(N^{-3i},N^{3i},\epsilon N^{-3i},N^{-3i})$ \\
\hline $(\bar{w}_{3i+1},a_{3i+1},b_{3i+2})$ & $(N^{-3i-1},N^{3i+1}, N^{-3i-1},1/\epsilon N^{-3i-1})$ \\
\hline $(\bar{w}_{3i+2},a_{3i+2},b_{3i})$ & $(N^{-3i-2},N^{3i+2}, N^{-3i-2},2N^{-3i-1})$ \\
 \hline
\end{tabular}
\label{table:machine}
\end{center}
\end{table}

\begin{table}[!hbp]
\begin{center}
\caption{Size Vectors of Jobs}
\begin{tabular}{|c|c|}
\hline Jobs & Sizes\\
\hline $w_i$ & $(N^{-i},N^{i},0,0)$\\
\hline $\bar{w}_i$ & $(N^{i},N^{-i},0,0)$\\
\hline $s_j$ & $(0,0,1/2N^{-j-N},1/2N^{j+N})$\\
\hline $s_j'$ & $(0,0,1/2N^{-j-N},1/2N^{j+N})$\\
\hline $a_i$ & $(N^{-i},N^{-i},\epsilon N^{i},0)$ \\
\hline $b_{3i}$ & $(\epsilon N^{-3i},\epsilon N^{-3i-2},1/2N^{3i},1/2N^{3i+1})$ \\
\hline $b_{3i+1}$ & $(\epsilon N^{-3i-1},1/\epsilon N^{-3i-1},1/(2\epsilon)N^{3i},1/2 N^{3i})$ \\
\hline $b_{3i+2}$ & $(\epsilon N^{-3i-2},\epsilon N^{-3i-1}, 1/2N^{3i+1},\epsilon/2 N^{3i+1})$ \\
 \hline
\end{tabular}
\label{table:job}
\end{center}
\end{table}

We check the processing times of jobs on different machines. The
following observation is easy to verify (by focusing on the first
three coordinates of vectors).
\begin{itemize}
\item A $w_i$-job (or $\bar{w}_i$-job) has a processing time of $2$
on a machine whose corresponding match contains $w_i$ (or
$\bar{w}_i$), and has a processing time of $\Omega(N)$ on other
machines.
\item An $s_j$-job ($s_j'$-job) has a processing time of $1$ on a
machine whose corresponding match contains $s_j$ ($s_j'$), and has a
processing time of $\Omega(N)$ on other machines.
\item An $a_i$-job has a processing time of $1+O(\epsilon)$ on a
machine whose corresponding match contains $a_i$, and has a
processing time of $\Omega(\epsilon N)$ on other machines.
\end{itemize}

For $b_i$-jobs, the reader may refer to Table~\ref{table:b-job} for
their processing times.

\begin{table}[!hbp]
\begin{center}
\caption{Processing times of $b_i$-jobs}
\begin{tabular}{|c|c|c|c|}
\hline Machines/Jobs & $b_{3i}$ & $b_{3i+1}$ & $b_{3i+2}$\\
\hline $(w_{3i},a_{3i},b_{3i})$ & $1+O(\epsilon)$ & $1/(2\epsilon)+O(\epsilon)$ & $\Omega(N)$\\
\hline $(w_{3i+1},a_{3i+1},b_{3i+1})$ & $\Omega(N)$ & $1+O(\epsilon)$ & $\Omega(\epsilon N)$\\
\hline $(w_{3i+2},a_{3i+2},b_{3i+2})$ & $\Omega(N)$ & $\Omega(\epsilon N)$ & $1+O(\epsilon)$\\
\hline $(\bar{w}_{3i},a_{3i},b_{3i+1})$ & $\Omega(N)$ & $1+O(\epsilon)$ & $\Omega(\epsilon N)$\\
\hline $(\bar{w}_{3i+1},a_{3i+1},b_{3i+2})$ & $\Omega(N)$ & $1/(2\epsilon)+O(\epsilon)$ & $1+O(\epsilon)$\\
\hline $(\bar{w}_{3i+2},a_{3i+2},b_{3i})$ & $1+O(\epsilon)$ & $\Omega(N)$ & $\Omega(\epsilon N)$\\
\hline Other machines & $\Omega(\epsilon N)$ & $\Omega(\epsilon N)$ & $\Omega(\epsilon N)$\\
 \hline
\end{tabular}
\label{table:b-job}
\end{center}
\end{table}

\noindent\textbf{Completeness.} Suppose $I_{3dm}$ admits a perfect
matching, we prove that $I_{sch}$ admits a feasible schedule of
makespan $2+O(\epsilon)$. We let $M_1$ be the set of machines
corresponding to the matches of the perfect matching, and $M_2$ be
the set of remaining machines. We provide a schedule in which all
the real jobs are put onto machines of $M_1$, while all the dummy
jobs are put onto machines of $M_2$. Since every element of $X\cup
Y$ appears once in the perfect matching, we put an $s_j$-job onto a
machine of $M_1$ whose corresponding match contains $s_j$.
$s_j'$-jobs, $a_i$-jobs and $b_i$-jobs are scheduled in the same
way. It is easy to see that the load of each machine in $M_1$ is
$2+O(\epsilon)$. Meanwhile, every $w_i$ (or $\bar{w}_i$) also
appears once in the perfect matching, thus it appears $d(w_i)-1$ (or
$d(\bar{w}_i)-1$) times in the remaining matches. Notice that the
number of $w_i$-jobs (or $\bar{w}_i$-jobs) is $d(w_i)-1$ (or
$d(\bar{w}_i)-1$), thus we can put one $w_i$-job (or
$\bar{w}_i$-job) onto a machine in $M_2$ whose corresponding match
contains $w_i$ (or $\bar{w}_i$), and again it is easy to see that
the load of each machine in $M_2$ is at most $2$.

\noindent\textbf{Soundness.} Suppose $I_{sch}$ admits a feasible
schedule of makespan strictly less than $3$, we prove that $I_{3dm}$
admits a perfect matching.

As $1/(2\epsilon)>3$, according to our discussion on the processing
times we know every element job should be on a machine whose
corresponding match contains this element. Notice that the
processing time of a $w_i$-job (or $\bar{w}_i$-job) is at least 2,
while the processing time of an $a_i$-job, $b_i$-job, $s_j$-job or
$s_j'$-job is at least 1, thus every $w_i$-job (or $\bar{w}_i$-job)
occupies one machine, and there are no other jobs on this machine.
Let $M_2$ be the set of machines where $w_i$-jobs and
$\bar{w}_i$-jobs are scheduled, and let $M_1$ be the set of
remaining machines, we show that the matches corresponding to
machines in $M_1$ forms a perfect matching. Let $T'$ be the set of
these matches. Notice that the number of $w_i$-jobs (or
$\bar{w}_i$-jobs) is $d(w_i)-1$ (or $d(\bar{w}_i)-1$), thus in $T'$
every $w_i$ (or $\bar{w}_i$) appears exactly once, which implies
that $|T'|=6n$. Furthermore, all the jobs corresponding to elements
of $X\cup Y$ are on machines of $M_1$, thus every element of $X\cup
Y$ appears at least once in $T'$. Notice that $|W|=|X|=|Y|=6n$,
since every element of $W\cup X\cup Y$ appears at least once in $T'$
and $|T'|=6n$, we conclude that every element of $W\cup X\cup Y$
appears exactly once in $T'$, implying that $T'$ is a perfect
matching.
\end{proof}

\section{APX-hardness for rank 3 scheduling}
We start with the one-in-three 3SAT problem. It is a variation of
the 3SAT problem. Precisely speaking, an input of the one-in-three
3SAT is a collection of clauses where each clause consists of
exactly three literals, and the problem is to determine whether
there exists a truth assignment of the variables such that each
clause is satisfied by exactly one literal (i.e., one literal is
true and two other literals are false).

It is proved in \cite{one-in-three} that the one-in-three 3SAT
problem is NP-complete.

Given an instance of the one-in-three 3SAT problem, say, $I_{sat}$,
we can apply Tovey's method to transform it into $I_{sat}'$ such
that

\begin{itemize}
\item Each clause of $I_{sat}$ contains two or three literals
\item Each variable appears three times in clauses, among the three occurrence there
are either two positive literals and one negative literal, or one
positive literal and two negative literals
\item There exists a truth assignment for $I_{sat}'$ where every
clause is satisfied by exactly one literal if and only if there is a
truth assignment for $I_{sat}$ where every clause is satisfied by
exactly one literal
\end{itemize}

The transformation is straightforward. For any variable $z$, if it only appears once in the clauses,
then we add a dummy clause as $(z\vee\neg z)$. Otherwise suppose it appears $d\ge 2$ times in the clauses,
then we replace its $d$ occurrences with $d$ new variables as $z_1$, $z_2$, $\cdots$, $z_d$, and meanwhile
add $d$ clauses as $(z_1\vee\neg z_2)$, $(z_2\vee\neg z_3)$, $\cdots$, $(z_{d}\vee\neg z_1)$ to enforce that
these new variables should take the same truth assignment. It is not difficult to verify that the constructed
instance satisfies the above requirements.

Let $\epsilon$ be an arbitrary small positive.
Throughout the following part of this section we assume that
$I_{sat}'$ contains $n$ variables and $m$ clauses, and let
$\xi=2^3$, $r=2^{10}\xi=2^{13}$, $N=n/\epsilon^2$.
We will construct a scheduling instance $I_{sch}$ in the following part of this section such that
if there exists a truth assignment for $I_{sat}'$ where every clause is satisfied by exactly one
literal, then $I_{sch}$ admits a feasible schedule whose makespan is $r+O(\epsilon)$. On the
other hand if $I_{sch}$ admits a feasible schedule whose makespan is strictly less
than $r+1$, then there exists a truth assignment for $I_{sat}'$ where every clause is satisfied by exactly one
literal. This would be enough to prove that an algorithm of approximation ratio strictly less than
$1+1/(r+1)=1+1/(2^{13}+1)$ implies that $P=NP$.

\subsection{Construction of the scheduling instance}
We construct jobs.

For every variable $z_i$, eight variable jobs are constructed,
namely $v_{i,k}^{\gamma}$ for $k=1,2,3,4$ and $\gamma=T,F$. The size
vectors are (for simplicity, we use $s(j)$ to denote the size vector
of job $j$):

$$s(v_{i,1}^T)=(\epsilon N^{4i+1},0,1/8r-10\xi-2), s(v_{i,2}^T)=(\epsilon N^{4i+2},0,1/8r-20\xi-2),$$
$$s(v_{i,3}^T)=(\epsilon N^{4i+3},0,1/8r-18\xi-2), s(v_{i,4}^T)=(\epsilon N^{4i+4},0,1/8r-12\xi-2).$$
$$s(v_{i,k}^F)=s(v_{i,k}^T)-(0,0,2), k=1,2,3,4$$

For every variable $z_i$, eight truth-assignment jobs jobs are
constructed, namely $a_i^{\gamma}$, $b_i^{\gamma}$, $c_i^{\gamma}$
$d_i^{\gamma}$ with $\gamma=T,F$. The size vectors are:

$$s(a_i^T)=(0,\epsilon N^i,2\xi+1), s(b_i^T)=(0,\epsilon N^i,4\xi+1),$$
$$s(c_i^T)=(0,\epsilon N^i,8\xi+1), s(d_i^T)=(0,\epsilon N^i,16\xi+1).$$
$$s(\tau_i^F)=s(\tau_i^T)+(0,0,1), \tau=a,b,c,d$$

For every clause $\beta_j$, if it contains two literals, then we
construct two clause jobs, namely $u_j^T$ and $u_j^F$. Otherwise it
contains three literals, and we construct three clause jobs, namely
one $u_j^T$ and two $u_j^F$. The size vectors are:

$$s(u_j^T)=(0,\epsilon N^{N+j}, 1/4r+2), s(u_j^F)=(0,\epsilon N^{N+j}, 1/4r+4).$$

We construct $2n-m$ true dummy jobs $\phi^T=(0,0,1/16r+2)$, and
$m-n$ false dummy jobs $\phi^F=(0,0,1/16r+4)$ (here it is not
difficult to verify that $n\le m$).

Finally we construct huge jobs. Indeed, there is a one-to-one
correspondence between huge jobs and machines. For ease of
description we first construct machines, and then construct those
huge jobs.

We construct $8n$ machines.

For every variable $z_i$, we construct $4n$ truth assignment
machines, and they are denoted as $(v_{i,1},a_i,c_i)$,
$(v_{i,2},b_i,d_i)$, $(v_{i,3},a_i,d_i)$, $(v_{i,4},b_i,c_i)$. The
symbol of a machine indicates the jobs on it (except the huge jobs)
in the solution with makespan at most $r+2\epsilon$. The speed
vectors are (For simplicity the speed vector of a machine, say,
$(v_{i,1},a_i,c_i)$, is denoted as $g(v_{i,1},a_i,c_i)$):

$$g(v_{i,1},a_i,c_i)=(N^{-4i-1},N^{-i},1),g(v_{i,2},b_i,d_i)=(N^{-4i-2},N^{-i},1),$$
$$g(v_{i,3},a_i,d_i)=(N^{-4i-3},N^{-i},1),g(v_{i,4},b_i,c_i)=(N^{-4i-4},N^{-i},1).$$

For every clause $\beta_j$, if the positive (or negative) literal $z_i$
(or $\neg z_i$) appears in it for the first time (i.e., it does not
appear in $\beta_k$ for $k<j$), then we construct a clause machine
$(v_{i,1},u_j)$ (or $(v_{i,3},u_j)$). Else if it appears for the
second time, then we construct a clause machine $(v_{i,2},u_j)$ (or
$(v_{i,4},u_j)$). The speed vectors are:
$$g(v_{i,k},u_j)=(N^{-4i-k},N^{-N-j},1).$$

Recall that for every variable, in all the clauses there are either
one positive literal and two negative literals, or two positive
literals and one negative literal. If $z_i$ appears once and $\neg
z_i$ appears twice, then we construct a dummy machine
$(v_{i,2},\phi)$, otherwise we construct a dummy machine
$(v_{i,4},\phi)$. The speed vectors are:
$$g(v_{i,2},\phi)=(N^{-4i-2},0,1), g(v_{i,4},\phi)=(N^{-4i-4},0,1).$$

According to our construction, it is not difficult to verify that if
$z_i$ appears once and $\neg z_i$ appears twice, then we construct
machines $(v_{i,k},u_{j_k})$ for $k=1,3,4$, $1\le j_k\le m$, and
machine $(v_{i,2},\phi)$. Otherwise we construct machines
$(v_{i,k},u_{j_k})$ for $k=1,2,3$, $1\le j_k\le m$, and machine
$(v_{i,4},\phi)$.

We now describe the huge jobs. There is one huge job for each
machine and for simplicity, we also use the symbol of a machine to
denote its corresponding huge job. The size vectors are:

$$s(v_{i,1},a_i,c_i)=(\epsilon N^{4i+1},\epsilon N^{i},7/8r),s(v_{i,2},b_i,d_i)=(\epsilon N^{4i+2},\epsilon N^{i},7/8r),$$
$$s(v_{i,3},a_i,d_i)=(\epsilon N^{4i+3},\epsilon N^{i},7/8r),s(v_{i,4},b_i,c_i)=(\epsilon N^{4i+4},\epsilon N^{i},7/8r).$$
$$s(v_{i,1},u_j)=(0,\epsilon N^{N+j},5/8r+10\xi), s(v_{i,2},u_j)=(0,\epsilon N^{N+j},5/8r+20\xi),$$
$$s(v_{i,3},u_j)=(0,\epsilon N^{N+j},5/8r+18\xi), s(v_{i,1},u_j)=(0,\epsilon N^{N+j},5/8r+12\xi).$$
$$s(v_{i,2},\phi)=(0,N^{2N},13/16r+20\xi), s(v_{i,4},\phi)=(0,N^{2N},13/16r+12\xi).$$

\subsection{From 3SAT to Scheduling}
Given a truth assignment of $I_{sat}'$, we schedule jobs according to Table~\ref{ta:overview}.
\begin{table}[!hbp]
\begin{center}
\caption{Overview of jobs}\label{ta:overview}
\begin{tabular}{|c|c|}
\hline machines & jobs\\
\hline
$(v_{i,1},a_i,c_i)$ & $v_{i,1}$, $a_i$, $c_i$, $(v_{i,1},a_i,c_i)$\\
\hline
$(v_{i,2},b_i,d_i)$ & $v_{i,2}$, $b_i$, $d_i$, $(v_{i,2},b_i,d_i)$\\
\hline
$(v_{i,3},a_i,d_i)$ & $v_{i,3}$, $a_i$, $d_i$, $(v_{i,3},a_i,d_i)$\\
\hline
$(v_{i,4},b_i,c_i)$ & $v_{i,4}$, $b_i$, $c_i$, $(v_{i,4},b_i,c_i)$\\
\hline
$(v_{i,k},u_j)$ & $v_{i,k}$, $u_j$, $(v_{i,k},u_j)$\\
\hline
$(v_{i,k},\phi)$ & $v_{i,k}$, $\phi$, $(v_{i,k},\phi)$\\
\hline
\end{tabular}
\end{center}
\end{table}

Recall that except for the huge jobs, the symbol of a job, say, $a_i$, may represent either $a_i^T$ or
$a_i^F$, we determine whether each job in the above table is true or false according to the truth assignment of
variables.

If variable $z_i$ is false,
then we schedule jobs on truth assignment machines as $(v_{i,1}^T,a_i^T,c_i^T)$,
$(v_{i,2}^T,b_i^T,d_i^T)$, $(v_{i,3}^F,a_i^F,d_i^F)$, $(v_{i,4}^F,a_i^F,c_i^F)$, otherwise we schedule jobs as $(v_{i,1}^F,a_i^F,c_i^F)$,
$(v_{i,2}^F,b_i^F,d_i^F)$, $(v_{i,3}^T,a_i^T,d_i^T)$, $(v_{i,4}^T,a_i^T,c_i^T)$.

Notice that every clause, say, $\beta_j$, is satisfied by exactly one literal. Suppose it contains three variables (the
argument is the same if it contains two literals), namely, $z_{i_1}$, $z_{i_2}$
and $z_{i_3}$ and is satisfied by the first variable.

Consider variable $z_{i_1}$. According to the construction of machines if $\beta_j$ contains
its positive literal then machine $(v_{i_1,k_1},u_j)$ is constructed for $k_1\in\{1,2\}$, and we schedule $u_j^T$ and $v_{i_1,k_1}^T$ on this machine.
This is possible since variable $z_{i_1}$ is true, and $v_{i_1,k_1}^T$ is thus not scheduled with truth assignment jobs. Similarly if $\beta_j$ contains
the negative literal $\neg z_{i_1}$, then the satisfaction of $\beta_j$ by variable $z_{i_1}$ implies that this variable is false. Furthermore, machine
$(v_{i_1,k_1},u_j)$ is constructed for $k_1\in\{3,4\}$ and again we schedule jobs $u_j^T$ and $v_{i_1,k_1}^T$ on this machine.

Consider variable $z_{i_2}$ (for variable $z_{i_3}$ the argument is the same). Again if $\beta_j$ contains
its positive literal then machine $(v_{i_2,k_2},u_j)$ is constructed for $k_2\in\{1,2\}$, and we schedule $u_j^F$ and $v_{i_1,k_1}^F$ on it. This is possible
since $\beta_j$ is not satisfied by literal $z_{i_2}$ and the variable $z_{i_2}$ is thus false, meaning that $v_{i_2,k_2}^F$ is not scheduled with truth assignment
jobs. Else if $\beta_j$ contains the negative literal $\neg z_{i_2}$, then machine $(v_{i_2,k_2},u_j)$ for $k_2\in\{3,4\}$ is
constructed and the variable $z_{i_2}$ is true, we schedule jobs $u_j^F$ and $v_{i_2,k_2}^F$ on this machine.

It is not difficult to verify that by scheduling in the above way, the load of every truth assignment machine and clause machine
is $r+O(\epsilon)$, and furthermore, for every $i$, 7 jobs out of $v_{i,k}^{\gamma}$ are scheduled on these machines where $k=1,2,3,4$ and $\gamma=T,F$. If the positive literal $z_i$ appears in clauses for once and $\neg z_i$ for twice, then the job $v_{i,2}$ is not scheduled. Otherwise if $z_i$ appears for twice while $\neg z_i$ for once, the job $v_{i,4}$ is left. These jobs are scheduled on
dummy machines according to Table~\ref{ta:overview}. Notice that there are in all $4n$ true variable jobs, among them $2n$ ones are on truth assignment machines, $m$ are on
clause machines (as $u_j^T$ is with a true variable job, and $u_j^F$ is with with a false one), thus $2n-m$ true ones are on dummy machines. Recall that there
are $2n-m$ true dummy jobs $\phi^T$ and $m-n$ false dummy jobs, we always schedule a true dummy job with a true variable job, and a false dummy job with a false
variable job. It is easy to see that in this way, the load of every dummy machine is $r+O(\epsilon)$.

Thus in all, if there exists a truth assignment for $I_{sat}'$ in which every clause is satisfied by exactly one literal, then there exists
a feasible schedule for $I_{sch}$ whose makespan is $r+O(\epsilon)$.

\subsection{From Scheduling to 3SAT}
The whole subsection is devoted to proving the following theorem.
\begin{theorem}
\label{th:sche to sat}
If there is a solution for $I_{sch}$ whose makespan is strictly less than
$r+1$, then there exists a truth assignment for $I_{sat}'$
where every clause is satisfied by exactly one literal.
\end{theorem}
To prove the theorem, we start with the
following simple observation.

\noindent\textbf{Observation: } The processing time of a job on
every machine is greater than or equal to the third coordinate of
its size vector.

Using the above observation, it is not difficult to calculate that the total processing time
of all the jobs is at least $8nr$.
Let $Sol^*$ be the solution whose makespan is strictly less than
$r+1$, then the load
of every machine is in $[r,r+1)$. We check the scheduling of jobs in this solution.

\begin{lemma}
In $Sol^*$, there is one huge job on each machine, furthermore
\begin{itemize}
\item A huge job corresponds to a dummy machine is on a dummy
machine
\item A huge job corresponds to a clause machine is on a clause
machine
\item A huge job corresponds to a truth assignment machine is on a
truth assignment machine
\end{itemize}
\end{lemma}
\begin{proof}
According to the observation, the processing time of a huge job is
at least $5/8r-20\xi>1/2r+1$, thus there is at most one huge job on
each machine. Given the fact that there are $8n$ machines and $8n$
huge jobs, we know that there is one huge job on each machine in
$Sol^*$.

Consider any huge job corresponding to a dummy machine. Notice that
the second coordinate of its size vector is always $N^{2N}$,
implying that the processing time of this job is $\Omega(N)$ on
clause machines and truth assignment machines, thus this job is on a
dummy machine. Recall that there are $n$ dummy machines and $n$ huge
jobs corresponding to dummy machines, these huge jobs must be on
these dummy machines, one for each.

Consider any huge job corresponding to a clause machine. The second
coordinate of its size vector is at least $N^{N}$, implying that its
processing time is at least $\Omega(N)$ if it is put on a truth
assignment machine. On the other hand it could not be put on a dummy
machine either, thus it must be on a clause machine.

Similar arguments show that a huge job corresponding to a truth
assignment machine must be on a truth assignment machine.

\end{proof}

\begin{lemma}
\label{le:variable-satisfy}
In $Sol^*$, except for the huge jobs,
\begin{itemize}
\item There is a variable job on each machine
\item There is a clause job on each clause machine
\item There is a dummy job on each dummy machine
\end{itemize}
\end{lemma}
\begin{proof}
Consider a clause job. Its processing time is greater than $1/4r$.
If it is put on a truth assignment machine, then the load of this
machine becomes larger than $1/4r+7/8r>r+1$, which is contradiction.
Else if it is put on a dummy machine, then the load of this machine
becomes larger than $1/4r+13/16r+20\xi>r+1$, which is also a contradiction. Hence a clause job could only
be on a clause machine. Meanwhile if there are two clause jobs on
one clause machine, then the load of this machine also becomes
larger than $5/8r+20\xi+1/2r>r+1$. As there are $n$ clause jobs and clause machine, there is exactly one clause
job on one clause machine.

Consider a variable job. It is not difficult to verify that there could not be two variable jobs on one machine since the
total processing time of two variable job is at least $1/4r-40\xi-8>3/16r$ (due to the fact that $r=2^{10}\xi$). Given that
there are $8n$ variable jobs and $8n$ machines, there is one variable job on each machine. Now a dummy job could only be
on a dummy machine, and similar arguments show that there could be at most one dummy job on a dummy machine, hence there is
one dummy job on one dummy machine.
\end{proof}

A machine is called variable-satisfied, if the variable job on this machine coincide with the symbol of this machine, i.e., for any machine
denoted as $(v_{i,k},*)$ or $(v_{i,k},*,*)$, the variable job on it is $v_{i,k}$ where $k=1,2,3,4$. We have the following lemma.

\begin{lemma}
Every machine is variable-satisfied.
\end{lemma}

\begin{proof}
Consider the eight jobs $v_{n,k}^{\gamma}$ where $\gamma=T,F$, $k=1,2,3,4$. For any machine denoted as $(v_{j,k},*)$ or $(v_{j,k},*,*)$,
the first coordinate of its speed vector is $N^{-4j-k}$, thus the processing time of $v_{n,k}$ on this machine becomes $\Omega(\epsilon N)$ if $j<n$.
Furthermore, it can be easily seen that $v_{n,4}$ could only be on machines denoted as $(v_{n,4},*)$ or $(v_{n,4},*,*)$. Since there are
two jobs $v_{n,4}$ (one true and one false), and two machines denoted as $(v_{n,4},*)$ or $(v_{n,4},*,*)$ (either machines $(v_{n,4},b_n,c_n)$ and $(v_{n,4},\phi)$, or
machines $(v_{n,4},b_n,c_n)$ and $(v_{i,4},u_{j_n})$ for some $j_n$), thus the two machines are satisfied. Iteratively applying the above
arguments we can prove that every machine is satisfied.
\end{proof}

Using similar arguments as the proof the above lemma, we can also prove that the huge job on every machine also coincide with the symbol of this machine.

A machine is called satisfied, if all the jobs on this machine coincide with the symbol of this machine, i.e., jobs are scheduled according
to Table~\ref{ta:overview}.

\begin{lemma}
Every machine is satisfied.
\end{lemma}
\begin{proof}
Notice that the second coordinate of a clause job $u_j$ is $\epsilon N^{N+j}$, and there is one clause job
on every clause machine, thus using similar arguments as the proof of Lemma~\ref{le:variable-satisfy}, we can show that
the clause job $u_j$ is on a the clause machine $(v_{i,k},u_j)$.
Now adding up the processing times of the huge job, clause job and variable job on a clause machine, the sum is at least $r-2$, implying that
there is no truth assignment jobs on clause machines, and thus every clause machine is satisfied.

Consider a dummy machine. According to Lemma~\ref{le:variable-satisfy}, the total processing time of a dummy job and a variable job on a machine
is at least $r-2$, thus again there is no truth assignment jobs on it and every dummy machine is satisfied.

Consider truth assignment machines. The above analysis implies that all the truth assignment jobs are on these machines.
We check machines $(v_{1,1},a_1,c_1)$,
$(v_{1,2},b_1,d_1)$, $(v_{1,3},a_1,d_1)$, $(v_{1,4},b_1,c_1)$. The total load of the four machines falls in $[4r,4r+4)$, and the
amount contributed by variable and huge jobs is among $[4r-60\xi-16,4r-60\xi-8]$, thus the amount contributed by truth assignment jobs is in $[60\xi+8,60\xi+20]$.
Notice that for any $i\ge 2$, the processing time of job $a_i$, $b_i$, $c_i$ or $d_i$ is at least $\Omega(\epsilon N)$ on the four machines we consider,
thus there are at most $8$ truth assignment jobs on these machines, namely $a_1$, $b_1$, $c_1$ and $d_1$. The total processing time of the $8$ jobs
is $60\xi+12$, while each of them has a processing time of at least $2\xi+1$, implying that all these jobs are on the four machines. Consider the two jobs $d_1$ (one
true and one false), either has a processing time at least $16\xi$, implying that they can only be on machine $(v_{1,2},b_1,d_1)$ and $(v_{1,3},a_1,d_1)$. Furthermore,
they can not be on the same machine, thus there is one $d_1$ on machine $(v_{1,2},b_1,d_1)$ and $(v_{1,3},a_1,d_1)$. Using the same argument we can prove that $a_1$
and $c_1$ are on machine $(v_{1,1},a_1,c_1)$, $b_1$ and $d_1$ are on machine $(v_{1,2},b_1,d_1)$, $a_1$ and $d_1$ are on machine $(v_{1,3},a_1,d_1)$, and $b_1$ and
$c_1$ are on machine $(v_{1,4},b_1,c_1)$. In all, the four machines $(v_{1,1},a_1,c_1)$,
$(v_{1,2},b_1,d_1)$, $(v_{1,3},a_1,d_1)$, $(v_{1,4},b_1,c_1)$ are all satisfied. Iteratively using the above arguments, we can prove that every truth assignment
machine is satisfied.

\end{proof}

Notice that except for a huge job, the symbol of a job, say, $a_i$, may represent either $a_i^T$ or $a_i^F$. A machine is called truth benevolent, if except the
huge job, all the jobs
on it are either all true or all false.

\begin{lemma}
Every machine is truth benevolent.
\end{lemma}
\begin{proof}
Consider a truth assignment machine, say, $(v_{i,1},a_i,c_i)$. If $v_{i,1}^T$ is on this machine, then $a_i$ and $c_i$ are both true, for otherwise
one of them is false, and the total processing time of the three jobs is at least $1/8r+1$, implying that the load of this machine is at least
$r+1$, which is a contradiction. Similarly, if $v_{i,1}^F$ is on this machine, then $a_i$ and $c_i$ are both false, for otherwise one of them is true,
and the total processing time of the three jobs plus the huge job is at most $r-1+O(\epsilon)<r$, which is a contradiction. Iteratively applying the above
arguments we can show that every truth assignment machine is truth benevolent.

Consider a clause machine, say, $(v_{i,k},u_j)$ for $k=1,2,3,4$. If $v_{i,k}^T$ is on this machine, then $u_j$ is true for otherwise the load of this machine is at least $r+2$, which
is a contradiction. If $v_{i,k}^F$ is on this machine, then $u_j$ is false, for otherwise the load of this machine is at most $r-2+O(\epsilon)<r$, which is also a contradiction.
Thus every clause machine is truth benevolent. Using the same argument we can also prove that every dummy machine is truth benevolent.
\end{proof}

Now we come to the proof of Theorem~\ref{th:sche to sat}. It is easy to see that for every $i$, jobs are either scheduled as $(v_{i,1}^T,a_i^T,c_i^T)$,
$(v_{i,2}^T,b_i^T,d_i^T)$, $(v_{i,3}^F,a_i^F,d_i^F)$, $(v_{i,4}^F,a_i^F,c_i^F)$ or $(v_{i,1}^F,a_i^F,c_i^F)$,
$(v_{i,2}^F,b_i^F,d_i^F)$, $(v_{i,3}^T,a_i^T,d_i^T)$, $(v_{i,4}^T,a_i^T,c_i^T)$. If the former case happens, we let
the variable $z_i$ be false, otherwise we let $z_i$ be true. We prove that, by assigning the truth value in this way, every
clause of $I_{sat}'$ is satisfied by exactly one literal.

Consider any clause, say, $\beta_j$, and suppose it contains three literals, say, $v_{i_1,k_1}$, $v_{i_2,k_2}$ and $v_{i_3,k_3}$ where $k_1,k_2,k_3\in\{1,2,3,4\}$.
Since there is one $u_j^T$ and two $u_j^F$, we assume that $u_j^T$ is scheduled with $v_{i_1,k_1}^T$.

We prove that $\beta_j$ is satisfied by variable $z_{i_1}$.
There are two possibilities. If $k_1\in\{1,2\}$, then $u_j^T$ and $v_{i_1,k_1}^T$ are on the machine $(v_{i_1,k_1},u_j)$, and according to the construction of
machines, such a machine is constructed as the positive literal $z_i$ appears in clause $\beta_j$ for the first or second time. According to our truth assignment,
variable $z_i$ is true, for otherwise $v_{i_1,k_1}^T$ is scheduled with $a_i^T$, $c_i^T$ or $b_i^T$, $d_i^T$. Otherwise $k_1\in\{3,4\}$, and machine $(v_{i_1,k_1},u_j)$
is constructed as the negative literal $\neg z_i$ appears in clause $\beta_j$ for the first or second time. Again according to the truth assignment now the variable
$z_i$ is false, thus in both cases $\beta_j$ is satisfied by variable $z_{i_1}$.

We prove that $\beta_j$ is not satisfied by either variable $z_{i_2}$ or $z_{i_3}$. Consider $z_{i_2}$, again there are two possibilities. If $k_2\in\{1,2\}$,
then the positive literal $z_{i_2}$ appears in $\beta_j$ for the first or second time, and meanwhile variable $z_{i_2}$ is false because otherwise $v_{i_2,k_2}^F$
is scheduled with $a_i^F$, $c_i^F$ or $b_i^F$, $d_i^F$, rather than $u_j^F$. Thus $\beta_j$ is not satisfied by variable $z_{i_2}$. Using the same argument we can
prove that if $k_2\in\{3,4\}$, $\beta_j$ is not satisfied by variable $z_{i_2}$, either. The proof is the same for variable $z_{i_3}$.

Thus in all, $\beta_j$ is satisfied by exactly one literal when it contains three literals. The same result also holds when $\beta_j$ contains two literals via the
same proof.

\bibliographystyle{plain}

\end{document}